\documentclass{article}
\usepackage{arxiv}

\usepackage{fancyhdr}
\usepackage[utf8]{inputenc}
\usepackage[numbers, sectionbib]{natbib}
\usepackage{amsmath, amssymb, amsfonts, amsthm}
\usepackage{mathtools}
\usepackage{graphicx}
\usepackage{color, xcolor}
\usepackage{float}
\usepackage{hyperref}
\usepackage{comment}
\usepackage{algorithm, algpseudocode}
\usepackage{enumerate}
\usepackage{enumitem}
\usepackage{subfigure}
\usepackage{times}
\usepackage{sidecap}
\usepackage{tcolorbox}
\usepackage{blindtext}
\usepackage[symbol]{footmisc}

\usepackage{boxedminipage}
\usepackage{xspace}

\theoremstyle{plain}
\newtheorem{theorem}{Theorem}[section]
\newtheorem{lemma}[theorem]{Lemma}

\theoremstyle{definition}

\newtheorem{observation}[theorem]{Observation}

\newtheorem{remark}[theorem]{Remark}

\newtheorem*{def*}{Definition}
\newtheorem*{prfthm*}{Proof of Theorem}

\newcommand{\remove}[1]{}

\AtBeginDocument{%
  }

\author{
 Prabhat Kumar Chand \\
  Indian Statistical Institute\\
  Kolkata,\\
  India. \\
  \texttt{pchand744@gmail.com} \\
   \And
 Manish Kumar \\
  Indian Institute of Technology\\
  Madras,\\
  India.\\
  \texttt{manishsky27@gmail.com} \\
  \And
 Anisur Rahaman Molla \\
  Indian Statistical Institute\\
  Kolkata,\\
  India. \\
  \texttt{molla@isical.ac.in} \\
}

\begin{document}

%\title{Agent-Based Triangle Counting and its Applications in Anonymous Graphs\thanks{An \textit{Extended Abstract} (2 pages) of this manuscript appeared in the proceedings of AAMAS 2024\cite{aamas_ea}}}

\title{Computing Tree Structures in Anonymous Graphs via Mobile Agents}

\maketitle
\sloppy 
\begin{abstract}
%\mki{The Abstract can be checked.}
Minimum Spanning Tree (MST) and Breadth-First Search (BFS) tree constructions are classical problems in distributed computing, typically studied in the message-passing model, where static nodes communicate via messages. This paper studies the MST and BFS tree in an agent-based network, in which computational devices are modeled as mobile agents that explore a graph and perform computations. Each node in the graph serves as a container for these mobile agents, and communication occurs between agents when they move to the same node. We consider the scenario where $n$ agents are dispersed (each node contains one agent) on the nodes of an anonymous, arbitrary $n$-node, $m$-edge graph $G$. The goal is for the agents to construct these tree structures from the dispersed configuration, ensuring that each tree-edge is recognized by at least one of the endpoint agents. The objective is to minimize the time to construct the trees and memory usage per agent. Following the literature we consider the synchronous setting where each agent performs its operation synchronously with others and hence the time complexity is measured in rounds. Crucially, we aim to design solutions where agents have no prior knowledge of any graph parameters such as $n, m, D, \Delta$, etc., where $D$ and $\Delta$ are the diameter and maximum degree of the graph. There exists a deterministic solution that constructs a BFS tree in $O(D\Delta)$ rounds with $O(\log n)$ bits of memory per agent in the agent-based network with a priori knowledge of the root. In this paper, we present a deterministic algorithm that constructs the BFS tree in $O(\text{min}(D\Delta, m\log n)+n\log n+\Delta \log^2 n)$ rounds with each agent using only $O(\log n)$ bits without any prior knowledge of the root. In the process of knowing the root, we solve two problems, namely, leader election and minimum spanning tree. We elect a leader agent (among the agents) and construct the MST in $O(n\log n+\Delta \log^2 n)$ rounds and $O(\log n)$ bits of memory. The prior known results for leader election are $O(m)$ rounds and $O(\log^2 n)$ bits of memory per agent while MST construction requires $O(m+n\log n)$ rounds and $O(\text{max}(\Delta, \log n) \log n)$ bits of memory per agent. Our results are a significant improvement over the prior results since memory usage is optimal and leader election and MST have time almost linear in $n$. In our algorithms, each agent is assumed to know the value of $\lambda$, defined as the maximum identifier among all agents. For analysis, we assume that $\lambda$ is bounded by a polynomial in $n$, i.e., $\lambda \leq n^c$ for some constant $c \geq 1$.

%\manish{Above is a bit detailed and tentative abstract, we may remove/add some of the parts.}
%This paper investigates the complexity of electing a leader and constructing a Minimum Spanning Tree (MST) and Breadth-First-Search (BFS) tree in an arbitrary, connected anonymous graph. The BFS construction methodology improves upon the BFS construction method presented in \cite{CKM24} by eliminating the need for agents to have a knowledge of a predetermined root node. In addition, we develop a algorithm using which the agents can meet its intended neighbor while performing independent tasks. Unlike previous approaches \cite{aamas_ea, chand23}, our method requires no knowledge of the graph’s maximum degree, $\Delta$, and instead operates solely with an upper bound on agent identifiers, denoted by $\Lambda$. This approach enables agents to dynamically choose and meet particular neighbors without needing to interact with all adjacent nodes. We use this new communication mechanism to facilitate both MST construction and leader election. The final BFS tree is then established with the elected leader agent serving as the root.

\end{abstract}

\section{Introduction}\label{sec: introduction}
Tree structures are fundamental in computer science, enabling efficient organization and retrieval across domains such as databases, networking, and AI. Among them, Breadth-First Search (BFS) Trees and Minimum Spanning Trees (MST) are widely used in both classical and distributed systems. BFS trees support shortest-path discovery in unweighted graphs~\cite{local_map15}, facilitating applications in networking, pathfinding, and distributed computing. MSTs and BFS trees are particularly useful in decentralized systems where mobile agents perform local computations, minimizing communication overhead—a key advantage in scenarios like disaster response, urban management, and military operations.  

Agent-based models have proven effective for memory-efficient graph exploration~\cite{local_map15}, subgraph enumeration~\cite{directed22}, and autonomous dispersion~\cite{kshemkalyani19,kshemkalyani22}. Mobile agents are widely applied in network exploration, from underwater navigation~\cite{CGZ2021} to network-centric warfare~\cite{LSP2018}, social network analysis~\cite{ZSW2018, social1}, and search-and-rescue missions~\cite{sar1, sar2, sar3, sar4}. In Wireless Sensor Networks (WSNs), drone-based agents enhance sensor deployment in remote areas~\cite{wsn0, wsn1, wsn2, wsn3}.  

Recent studies explore agent-based solutions for maximal independent sets (MIS)~\cite{pramanick2023filling, pattanayak24}, small dominating sets~\cite{chand23}, triangle counting~\cite{aamas_ea}, BFS tree construction~\cite{CKM24}, and message-passing simulations~\cite{KKMS24}. This paper uses a mechanism for agents to meet their neighbours by exploiting their ID bits without relying on graph parameters, enabling improved MST construction over~\cite{KKMS24} and enhancing BFS tree formation from~\cite{CKM24}. Techniques that leverage agent identifiers for coordination and traversal have indeed been explored in the literature, particularly in deterministic rendezvous and treasure hunt problems using strongly universal exploration sequences (UXS)~\cite{talg_14}. While our proposed protocol shares the high-level idea of employing agent labels, it differs in both the problem setting and the manner in which labels are used. Specifically, our protocol focuses on enabling agents to meet their neighbors within a designated round window in an anonymous graph, using labels to coordinate local interactions rather than global exploration.

\subsection{Our Contribution}

%\pci{May be checked!}
%\mki{Working.}
Recently, Chand {\em et al.}~\cite{CKM24} studied BFS tree construction with various relevant problems in the agent-based model. Their main result shows a deterministic BFS tree construction in $O(D\Delta)$ rounds and uses $O(\log n)$ bits of memory at each agent, where $D$ and $\Delta$ are the diameters and maximum degree of the graph $G$, respectively. Kshemkalyani {\it et al.}~\cite{KKMS24} recently provided the first deterministic algorithm for leader election in the agent-based model, which elects the leader in $O(m)$ time with $O(n\log n)$ bits at each agent. With the help of the elected leader, they constructed the MST in $O(m+n \log n)$ time with $O(n\log n)$ bits at each agent. The algorithm achieves so without relying on any knowledge (neither exact nor an upper bound) on graph parameters, such as $n$ (the network size and also the number of agents), $\Delta$ (the maximum degree of $G$), and $D$ (diameter of $G$). Similarly, the algorithms in this paper do not require knowledge of any graph parameters such as $n$,$\Delta$, etc. Each agent only knows the value of $ \lambda $, a parameter that depends solely on the agent identifiers and is independent of the graph structure. For analysis purposes, we assume that agent IDs are bounded by $ n^c $ for some constant $ c \geq 1 $, which ensures that each ID can be represented using $ O(\log n) $ bits. This allows the algorithms to operate without requiring agents to have prior knowledge of $ n $. The graph $ G $ is anonymous—only agents have unique IDs. In this paper, along the line of work of~\cite{CKM24} and \cite{KKMS24}, we prove the following theorems.

%\manish{I have mentioned all the theorems here. We may remove them from the other sections and refer from here.}

\textbf{Theorem~\ref{theorem: MST} (MST)}
     Given a dispersed configuration of $n$ agents with unique identifiers positioned one agent at every node of $n$ nodes, $m$ edges, maximum degree $\Delta$, and a graph $G$ with no node identifier. Then, a deterministic algorithm constructs an MST in $O(\Delta \log^2 n+n\log n)$ rounds and $O(\log n)$ bits per agent. The algorithm requires no prior knowledge of any graph parameters and assumes that each agent knows only $\lambda$, the maximum among all agent identifiers.
     {\bf (Section~\ref{sec: MST_construction})}

Theorem \ref{theorem: MST} is a significant improvement over the only previously known result for MST construction in the agent-based model due to Kshemkalyani {\it et al.} \cite{KKMS24}. In fact,  the memory complexity is optimal since any algorithm designed in the agent-based model with $n$ agents needs $\Theta(\log n)$ bits per agent~\cite{KshemkalyaniMS19,Augustine:2018} and the time complexity has improved from $O(m+n\log n)$ rounds. Besides its merits regarding improved time and memory complexities for an important problem, the root of the MST provides an immediate result for leader election, and we have the following theorem.

\textbf{Theorem~\ref{theorem: Leader} (Leader Election)}
     Given a dispersed configuration of $n$ agents with unique identifiers positioned one agent at every node of $n$ nodes $m$ edges, maximum degree $\Delta$ graph $G$ with no node identifier. Then, there is a deterministic algorithm that elects a leader in $O(\Delta \log^2 n+n\log n)$ rounds and $O(\log n)$ bits per agent. The algorithm requires no prior knowledge of any graph parameters and assumes that each agent knows only $\lambda$, the maximum among all agent identifiers.
 {\bf (Section~\ref{sec: MST_construction})}

The elected leader, with the help of the MST algorithm, eventually works as the root of the BFS tree. With the help of the root, we developed an algorithm with the following theorem.

\textbf{Theorem~\ref{theorem: BFS} (BFS Tree)}
    Given a dispersed configuration of $n$ agents with unique identifiers positioned one agent at every node of $n$ nodes $m$ edges, maximum degree $\Delta$, and diameter $D$ graph $G$ with no node identifier. Then, there is a deterministic algorithm that constructs a BFS tree in $O(\min(D\Delta, m \log n) + n \log n + \Delta \log^2 n)$ rounds and $O(\log n)$ bits per agent. The algorithm requires no prior knowledge of any graph parameters and assumes that each agent knows only $\lambda$, the maximum among all agent identifiers.
    {\bf (Section~\ref{sec: BFS_improved})}

Theorem~\ref{theorem: BFS} enhances the results for graphs with a high diameter ($D$) and a high maximum degree ($\Delta$), which are particularly significant as they balance long-range connectivity with local hub efficiency. These properties make such graphs valuable in various domains, including network design and communication, where they contribute to scalability and delay-tolerant networks; social networks, where they influence information spread; and computational applications, where they present challenges in graph traversal algorithms.

To establish the above theorems, we introduced a protocol, termed \texttt{Meeting with Neighbor()} (details in Section~\ref{sec: Procedure_MN}), to systematically enable agents to meet each of their neighbors within $O(\log n)$ time per neighbor. This protocol operates without requiring any graph-specific parameters and only needs knowledge of the ID of the highest ID agent and therefore allows computation between neighboring agents without using any knowledge of graph parameters, unlike previous approaches~\cite{chand23,aamas_ea,CKM24}. Using \texttt{Meeting with Neighbor()}, we construct a BFS tree in an arbitrary anonymous graph. To facilitate this, agents first identify a root and gather essential graph parameters such as $m$ and $\Delta$ by constructing a Minimum Spanning Tree, after which the Breadth First Search Tree is constructed.  Table~\ref{tbl: comparative_analysis} lists and compares all the established results with the previous results. 

While our graph $G$ is unweighted, we adapt our approach from the GHS algorithm due to its inherently parallel structure, which offers a distinct advantage over the existing agent-based MST approach in~\cite{KKMS24} (See Table~\ref{tbl: comparative_analysis}). Although \cite{KKMS24} also utilizes a variant of the GHS algorithm, it operates sequentially. In addition to an independent MST construction study using agents, the GHS algorithm allows for the construction of a spanning tree with a root. This is therefore advantageous than the spanning tree construction method used in~\cite{CKM24}.

\subsection{Challenges, Techniques and High-Level Ideas}

In most agent-based graph algorithms, particularly in~\cite{chand23,pattanayak24,aamas_ea}, agents need to know $\Delta$, the maximum degree of the graph, to coordinate meetings effectively between adjacent agents, especially when operating autonomously. This knowledge of $\Delta$ allows agents to create a window in some multiples of $\Delta$ rounds, ensuring that agents meet with all their neighbors within this timeframe. However, our approach overcomes this requirement by employing a new ID system, where agents concatenate a copy of their complemented padded ID bits. This mechanism guarantees that an agent $r_i$ seeking to meet its neighbor $r_j$ can do so without needing agents to know $\Delta$. This methodology allows an agent $r_u$ to meet any neighboring agent within $O(\log n)$ rounds. The agents only require prior knowledge of the highest ID among the agents to get a bound on the ID length of the $n$ agents. The maximum ID across the $n$ agents is denoted as $\lambda$, so the IDs of the agents fit within a $\log\lambda$ length. 

To construct a spanning tree, we adapt the distributed GHS Minimum Spanning Tree (MST) algorithm for our mobile agent model~\cite{ghs83}. The first challenge is assigning distinct weights to the graph’s edges to ensure the existence of a unique minimum spanning tree (a requirement for the GHS algorithm). To assign a distinct weight to an adjacent edge, the agents use their ID and port numbering of the outgoing edge. Each agent computes the weight of its adjacent edges deterministically whenever it is required (without storing these weights)  based on the IDs of the agents across the edge and their port numbers. The spanning tree construction works similarly to the GHS algorithm, where initially there are $n$ separate tree fragments. In each phase, these fragments compute the Minimum Weight Outgoing Edge (MWOE) and use it to merge or absorb other fragments. The final spanning tree is achieved when no further MWOE exists, and all nodes belong to a single fragment. A critical challenge arises during the MWOE selection. An agent who discovers a new outgoing edge must send this information to its fragment leader via the tree so that the leader can calculate the minimum weight. The leader can then gather all the edge weights and send the MWOE throughout the tree. However, this process requires the agents to move sequentially to their children to propagate this information. To manage this, agents keep moving to and fro between their home nodes and parent nodes periodically, allowing those at the same level to collect the MWOE from their parents in a single round. Once the MWOE is identified, the agents use these edges to merge different fragments.  Upon completing the spanning tree construction, the leader of the final fragment, designated as $r^\star$, becomes the leader of all the $n$ agents. Through the MST, agents also gather the necessary graph parameters. The BFS construction now begins from the leader as the root.  
% (see details in Algorithm~\ref{MWOE Selection}), (see details in Algorithm~\ref{MWOE Calculation by Leader})
The BFS construction can be completed in $O(D\Delta)$ rounds as demonstrated in~\cite{CKM24}. However, we introduce another approach that completes the BFS construction in $O(m \log n)$ rounds by utilizing the MST built in the previous step. This MST enables the agents to gather important global parameters, such as $m$ and $\Delta$. By integrating both BFS approaches, our algorithm achieves an overall round complexity of $O(\min(D\Delta, m \log n))$ for BFS construction, offering improved time complexity depending on the graph's structure.\\

\section{Model and Definition}\label{sec:model}

%\mki{This Section can be checked.}
%\subsection{Model}
\textbf{Graph: }We have an underlying graph $G(V,E)$ that is connected, undirected, unweighted and anonymous with $|V| = n$ nodes and $|E| = m$ edges. Nodes of $G$ do not have any distinguishing identifiers or labels. These nodes do not possess any memory and hence cannot store any information. The degree of a node $v\in V$ is denoted by $\delta(v)$ and the maximum degree of $G$ is $\Delta$. Edges incident on $v$ are locally labeled using port numbers in the range $[0,\delta(v)-1]$. We denote a port number that connects the nodes $u$ to $v$ by $p_{uv}$. $p_{uv}$ represents the edge with the outgoing port number of node $u$ and the incoming port number of node $v$. In general, $p_{uv}\neq p_{vu}$. The edges of the graph serve as \emph{routes} through which the agents can commute. Any number of agents can travel through an edge at any given time. We also consider a weighted graph for the MST problem, where the edge weights are positive numbers bounded by $O(n^c)$ for some constant $c\geq 1$. The edge weights are associated with the port numbers of the respective edges. An agent at a node $v$ can see the weights of the edges adjacent to $v$.\\

\noindent\textbf{Mobile Agents: }We have a collection of $n$ agents enumerated as $\mathcal{R} = \{r_1,r_2,\dots,r_n\}$ residing on the nodes of the graph with each having a unique ID $\in$ $[0,n^c]$ for some constant $c\geq 1$, and has $O(\log n)$ bits of memory to store information. We assume that the highest ID among the $n$ agents is known to all agents, denoted by $\lambda$ with ($\lambda\leq n^c$)\footnote{We use $O(\log n)$ instead of $O(\log \lambda)$ throughout the paper in theorem and result formation to ensure the results are expressed using graph parameters. This also helps to show the direct comparison with the existing results.}. An agent retains and updates its memory as needed. Two or more agents can be present (\emph{co-located}) at a node or pass through an edge in $G$. However, an agent is not allowed to stay on an edge. An agent can recognize the port number through which it has entered and exited a node. The agents do not have any visibility beyond their (current) location at a node. An agent at a node $v$ can only realize its adjacent ports (connecting to edges) at $v$. Only the collocated agents at a node can sense each other and exchange information. An agent can exchange all the information stored in its memory instantaneously. We consider that $n$ agents start from a complete {\em dispersed initial configuration}\footnote{For other starting configurations of agents, dispersion is first achieved in $O(n\log^2n)$ rounds using ~\cite{sudo24}.} such that each node of the graph $G$ has exactly one agent.  \\\\

\noindent \textbf{Communication Model: }We consider a synchronous system where the agents are synchronized to a common clock and the {\em local communication} model, where, only co-located agents (i.e., agents at the same node) can communicate among themselves. In each round, an agent $r_i$ performs the $Communicate-Compute-Move$ $(CCM)$ task-cycle as follows: (i) {\em Communicate:} $r_i$ may communicate with other agents at the same node, (ii) {\em Compute:} Based on the gathered information and subsequent computations, $r_i$ may perform all manner of computations within the bounds of its memory, and (iii) {\em Move:} $r_i$ may move to a neighboring node using the computed exit port. We measure the complexity in two metrics, namely, time/round and memory. The {\em time complexity} of an algorithm is the number of rounds required to execute the algorithm. The {\em memory complexity} is measured w.r.t. the amount of memory (in bits) required by each agent for computation.

\subsection{Problem Statements}
Consider an $n$-node simple, arbitrary, connected, and anonymous graph $G(V,E)$. Assume that $n$ autonomous agents are initially placed on the $n$ nodes of the graph, with exactly one agent per node (\emph{dispersed initial configuration}). The objective is to solve the following problems and analyze their computational complexity:  

\begin{enumerate}
    \item \emph{Construction of a Minimum Spanning Tree (MST):} To construct an MST of $G$ such that each edge belonging to the MST is known to at least one of its endpoint agents.
    
    \item \emph{Leader Election:} To designate a unique leader among the $n$ mobile agents.
    
    \item \emph{Construction of a Breadth-First Search (BFS) Tree:} To construct a BFS tree of $G$ such that each edge in the BFS tree is known to at least one of its endpoint agents.
\end{enumerate}

% \begin{definition}
%     Consider an arbitrary, anonymous, simple, connected graph $G$ with $n$ nodes. We analyze the complexity of constructing a Minimum Spanning Tree (MST) followed by a Breadth-First Search (BFS) tree in $G$ using $n$ mobile agents. These agents begin in a dispersed initial configuration and have knowledge only of $\lambda$, an upper bound on the ID length of the $n$ agents.
% \end{definition}

\section{Related Work}\label{sec: related_work}

%\mki{This section can be checked.}
\begin{table}[!t]
\centering \footnotesize
\begin{tabular}
%{|P{4cm}|P{12cm}|}
%{|P{3.45cm}|P{2.00cm}|P{5.3cm}|P{3.25cm}|}
{|p{3.5cm}|p{1.75cm}|p{3.25cm}|p{3.25cm}|}

%\hline
%\multicolumn{4}{|c|} {\bf Comparing the existing and developed results}\\
\hline
 {\bf Algorithm} & {\bf Knowledge} & {\bf Time} & {\bf Memory/agent}\\
\hline
%&&&\\
\multicolumn{4}{|c|} {\bf Leader Election}\\
\hline
%\textbf{Leader Election}&&&\\
Section~\ref{sec: MST_construction}& $\lambda$ & $O(n \log n+ \Delta \log^2 n)$ & $O(\log n)$\\
%&&&\\
Kshemkalyani {\it et al.}~\cite{KKMS24} & $-$ & $O(m)$ & $O(\log^2 n)$\\
%&&&\\
%\hline
%\arrayrulecolor{lightgray}
\hline
%\arrayrulecolor{black}
\multicolumn{4}{|c|} {\bf MST}\\
\hline
Section~\ref{sec: MST_construction}&$ \lambda $& $O(n \log n+ \Delta \log^2 n)$& $O(\log n)$\\
%&&&\\
Kshemkalyani {\it et al.}~\cite{KKMS24}&$-$& $O(m + n \log n)$& $O(\text{max}(\Delta, \log n)\log n)$\\
%&&&\\
%\arrayrulecolor{lightgray}
\hline
%\arrayrulecolor{black}
\multicolumn{4}{|c|} {\bf BFS}\\
\hline
Section~\ref{sec: BFS_improved}&$\lambda$&$O(\text{min} (m \log n, D \Delta)+ n \log n +\Delta \log^2 n)$& $O(\log n)$\\
Section~\ref{sec: BFS_improved}&root&$O(\text{min} (m \log n, D \Delta))$& $O(\log n)$\\
%&&&\\
Chand {\it et al.}~\cite{CKM24} & root & $O(D \Delta)$ & $O(\log n)$\\

\hline
\end{tabular}
\caption{Comparing previous and developed results for three problems in a graph $G$ in the agent-based model. `$-$' means no a priori knowledge of root and other parameters. $\lambda$ is the highest ID agent in an $n$-node graph with $m$ edges, diameter $D$, and maximum degree $\Delta$.
%\manish{I'm not sure where are we introducing $\lambda$ or we have introduced it earlier. Could you please check?}\prabhat{some reviewers argue that since $\lambda$ needs to be known, it must be reflected in the results table. that's why I changed it. $\lambda$ is introduced in the model. }
}\label{tbl: comparative_analysis}

\end{table}

% Trees are one of the most important sub-structures of a graph. Different sub-tree structures of a graph are useful for various different purposes and therefore has been extensively studied in both classical and distributed computing literature. 

The leader election problem was first stated by Le Lann \cite{Lann77} in the context of token ring networks, and since then it has been central to the theory of distributed computing. Awerbuch \cite{Awerbuch87} provided a deterministic algorithm with time complexity $O(n)$ and message complexity $O(m+n\log n)$. Peleg \cite{Peleg90L} provided a deterministic algorithm with optimal time complexity $O(D)$ and message complexity $O(mD)$. Recently, an algorithm is given in \cite{KPP0T15} with message complexity $O(m)$ but no bound on time complexity, and another algorithm with $O(D\log n)$ time complexity and $O(m\log n)$ message complexity. Additionally, it was shown in \cite{KPP0T15} that the message complexity lower bound is $\Omega(m)$ and time complexity lower bound is $\Omega(D)$ for deterministic leader election in graphs.  Leader election was not studied in the agent-based model before, except for \cite{KKMS24}. %See Table~\ref{tbl: comparative_analysis}. \anis{check this sudden "See Table~1". No matching with the previous sentences.}

One of the earliest studies of constructing a Minimum Spanning Tree (MST) in a distributed setting was made by Spira~\cite{s77}, focusing primarily on the communication complexity of the algorithm. This study laid the groundwork for the classical GHS algorithm~\cite{ghs83} for constructing MST in distributed networks. In the paper, the authors presented algorithms for constructing an MST in $O(n\log n)$ rounds using $O(n\log n + m)$ messages, where each message contained at most one weight value—a non-negative integer less than $\log n$, along with three additional bits. Later, the time complexity was reduced to $O(n)$ rounds in subsequent work~\cite{a87} and further to $O(\sqrt{n} \log^\star n + D)$ in~\cite{gkp98}. Additionally,~\cite{pr99} established a lower bound of $O(\frac{\sqrt{n}}{\log{n}} + D)$ rounds for the time complexity of distributed spanning tree construction. More recently, MST construction for an agent-based model was introduced in~\cite{KKMS24}. This study demonstrated an algorithm to build an MST with $n$ agents, achieving a time complexity of $O(m + n \log n)$ rounds. For agents beginning in a dispersed initial configuration, the algorithm requires $O(\max(\Delta, \log n)\log n)$ bits of memory per agent; however, if agents start from arbitrary positions, the memory requirement per agent increases to $O(n \log n)$ bits.
The earliest agent-based approach to BFS tree construction is found in the work of Awerbuch~\cite{awerbuch95, awerbuch99}, where a single agent explores the graph via a BFS traversal. To manage resource limitations, this agent periodically returns to its starting node, a technique known as \textit{Piecemeal Exploration} (for instance, to recharge). In more recent work, Palanisamy et al.~\cite{palanisamy20} proposed a BFS traversal method for a multi-agent distributed system by dividing the graph into clusters, assigning an agent to each cluster to construct parts of the BFS tree in parallel. BFS traversal has also been applied to solve the dispersion problem. For example, Kshemkalyani et al.~\cite{kshemkalyani22} used a BFS-based strategy to disperse agents across the graph, ensuring nearly one agent per node. Under the local communication model—where agents can only communicate when co-located—their BFS-based algorithm achieves a time complexity of $O(D\Delta(D + \Delta))$ rounds when all agents start from a single node. In the case where agents start from arbitrary nodes, their BFS-based dispersion requires $O((D + k)\Delta(D + \Delta))$ rounds, but this assumes a global communication model, allowing agents to communicate at any time, regardless of location. Both algorithms require each agent to have $O(\log D + \Delta \log k)$ bits of memory. Most recently, in~\cite{CKM24}, $n$ agents were employed to construct a BFS tree in an arbitrary, anonymous $n$-node graph. Starting with a root node, the agents first built a spanning tree to determine $\Delta$, which was then used to complete the BFS tree. This algorithm completes BFS construction in $O(D\Delta)$ rounds when the root node is initially known. Additional related work on distributed BFS construction is available in~\cite{CKM24}. In~\cite{manish_icdcit}, the authors addressed leader election and MST construction in a similar model, with the main focus on minimizing time and memory when the graph parameters $\Delta$ and $n$ are known to the agents a priori. They gave a deterministic leader election algorithm running in $O(n \log^2 n + D \Delta \log n)$ rounds with $O(\log n)$ bits per agent, and an MST construction in $O(m + n \log n)$ rounds using $O(\Delta \log n)$ bits per agent.

\section{Agent-Based MST Construction}\label{sec: MST_construction}

We construct the Breadth-First-Search (BFS) tree from an initial dispersed configuration of the $n$-agents on the graph $G$ having $n$-nodes. To construct a BFS tree, a root node is needed. To achieve this, we first elect a leader that serves as the root. The Minimum Spanning Tree (MST) procedure provides a leader election as a subroutine which becomes the root. Therefore, first, we construct an MST and elect the leader; in the process, we also gather some graph parameters that are required for the BFS tree construction. The MST is constructed by adapting the existing distributed MST algorithm developed by GHS~\cite{ghs83}. Although our graph $G$ is unweighted we adapt our approach from the GHS algorithm due to its parallel nature. Therefore, it also improves the existing agent-based MST result~\cite{KKMS24} which also uses the variant of the GHS algorithm but in a somewhat sequential way.

To construct the MST from our unweighted graph $G$, we first assign distinct edge weights to every edge (details in Section~\ref{sec: unweighted_spanning_tree}). For MST construction, a meeting of two adjacent agents is an essential condition. The challenging part is synchronizing the movement of the agents. As a solution, we develop an approach that requires $O(\log n)$ rounds to meet two adjacent agents across the same edge (in Section~\ref{sec: Procedure_MN}).
%To construct the BFS tree, we employ the following approach. First, we construct a spanning tree to gather essential graph parameters: $\Delta$, $m$, and $n$. This spanning tree is built using a modified version of the distributed GHS Minimum Spanning Tree (MST) algorithm. Although our graph $G$ is unweighted, we use the GHS algorithm for two key reasons. First, it enables a highly distributed construction by allowing different fragments of the tree to be formed in parallel, which are then merged simultaneously. This approach contrasts with the spanning tree algorithm in~\cite{CKM24}, which requires a designated root to initiate the process and runs in $O(D\Delta)$ rounds. Second, this approach provides an independent examination of MST construction for weighted graphs and offers improvements over previous agent-based MST methods.

\subsection{Meeting with Neighboring Agents}\label{sec: Procedure_MN}

 The agents operate autonomously without any centralized control, so, ensuring that an agent successfully meets its intended neighbor is challenging. This difficulty arises because the targeted neighbor might simultaneously move to meet another of its own neighbors in the same round. Therefore, a mechanism is needed to guarantee that two adjacent agents, intending to meet, can do so within a specified time frame. One such mechanism, known as \texttt{Protocol MYN}, has been employed in several studies~\cite{chand23,aamas_ea,CKM24}, where agents use the knowledge of $\Delta$ (the maximum degree) along with their ID bits to ensure that each agent meets all its neighbors at least once within $O(\Delta\log\lambda)$ rounds. However, this approach has two key drawbacks. First, agents must have prior knowledge of $\Delta$, the highest degree in the graph. Secondly, the order in which neighbors are met is dictated by the ID bits rather than the node’s port numbers. In some algorithms, agents may be required to meet only a specific neighbor or visit the neighbors in a sequence of port numbers. Our proposed approach overcomes both of these limitations. 

 Our proposed algorithm \texttt{Meeting with Neighbor()} (Algorithm~\ref{alg:meeting_protocol}) is as follows. Consider two agents, $r_u$ and $r_v$, initially located at adjacent nodes $u$ and $v$, respectively. Let $p_u$ and $p_v$ denote the ports connecting these nodes. Each agent $r_u$ first pads its ID $r_u.ID$ with leading $0$s to ensure a binary length of exactly $\log \lambda$ bits. Let us call this padded ID as $r_u.b$ (Line 1). It then computes a new identifier, $new\_ID$, by appending the bitwise complement of its ID to the left (Line 2), resulting in a $2\log \lambda$-bit string. The protocol runs for {$4\log \lambda$ rounds, with two rounds corresponding to a bit position (Lines 3-10). Agents scan their $new\_ID$ from the least significant bit (LSB) to the most significant bit (MSB). In each of these rounds, if the current bit of $r_u.new\_ID$ ($new\_ID$ of agent $r_u$) is $1$, $r_u$ moves to neighboring node $v$; otherwise, it stays at node $u$. Since the agents have distinct IDs, there must exist a bit position in $r_u.new\_ID$ and $r_v.new\_ID$ where the former has a 1 and the latter has a 0. This guarantees that $r_u$ moves to $v$ and encounters $r_v$.

For example, let $r_u$ have ID $100$ and $\lambda = 1100$. Padding gives $0100$, and computing $new\_ID$ results in $10110100$. If $r_v$ has ID $110$, then $r_v.new\_ID$ is $10010110$. In this case, $r_u$ meets $r_v$ in the $2 \cdot 6=12$th round. %\hl{We see that padding of an ID only changes the representation of the ID not the value of the ID, therefore, in the rest of the paper we use the word ID for the padded ID until and unless specified.}

\begin{algorithm}[t]
    \caption{Meeting with Neighbor()}
    \label{alg:meeting_protocol}
    \begin{algorithmic}[1]
        \Require Two agents $r_u$ and $r_v$ with unique IDs, located at adjacent nodes $u$ and $v$.
        \Ensure The agent $r_u$ meets $r_v$ at $v$ within $4\log \lambda$ rounds.
        \State Pad each agent's ID to $\log \lambda$ bits by adding leading zeros and call it $b$.
        \State $r_u.new\_ID\leftarrow (\sim r_u.b)||(r_u.b)$.
        \State Let us suppose $round\_no$ represents the round number initialized to $0$.
        \For{$ 4 \log \lambda$ rounds}
            \If{$round\_no\; mod\; 2 =0 $}
                \State $r_u$ scans $r_u.new\_ID$ from LSB to MSB.
                \If{$(round\_no/2 +1)^{th}$ bit of $r_u.new\_ID$ is $1$} 
                    \State $r_u$ moves to node $v$ via port $p_u$ and return to node $u$
                \Else 
                    \State $r_u$ remains at node $u$ for 2 consecutive rounds
                \EndIf
            \EndIf
            %\State $round\_no \leftarrow round\_no+1$
        \EndFor
    \end{algorithmic}
\end{algorithm}

\begin{observation}\label{lemma:meeting_protocol}
Let $r_u$ be an agent located at node $u$ with degree $\delta_u$. Let $v$ be a neighboring node of $u$ hosting an agent $r_v$. Then $r_u$ can meet $r_v$ within $4\log \lambda$ rounds. Furthermore, $r_u$ can meet all its neighbors in at most $\delta_u \cdot 4\log \lambda$ rounds.
\end{observation}
%\manish{I don't see we need to prove it formally the above Lemma neither we are doing. We may keep the ``Lemma 1" as ``Observation 1".}

    Observe that concatenating the IDs with their bitwise complement is crucial; otherwise, the two agents may not meet. For instance, consider agents $r_u$ and $r_v$ with IDs $0010$ and $0110$, respectively. Suppose $r_u$ aims to meet $r_v$, but $r_v$ is simultaneously engaged in another meeting with a neighbor $r_w \neq r_u$. In the second round, both $r_u$ and $r_v$ move, preventing their meeting. In the subsequent rounds, $r_u$ remains stationary, failing to meet $r_v$. However, using $new\_ID$, this issue is avoided. The new IDs for $r_u$ and $r_v$ are $11010010$ and $10010110$, respectively. In this case, $r_u$ is guaranteed to meet $r_v$ in the $7$th round, where $r_u$ has a $1$ and $r_v$ has a $0$.  
Having established the meeting with neighboring agents, we now move to the next section, where we construct a minimum spanning tree of a graph using agents. 
% We assume that the $ n $ agents are initially positioned with one agent per node. In this section, we construct two spanning trees. First, we generate a spanning tree for the unweighted graph, which will later be utilized in the BFS construction. Second, we construct a mini\pci{This Section can be checked.}mum spanning tree (MST) for the case when the graph is weighted (see Model~\ref{sec:model}).
 
\subsection{Constructing a Spanning Tree on an Unweighted Graph}\label{sec: unweighted_spanning_tree}
%We now present an algorithm for the agents to construct a spanning tree for a graph $G(V,E)$. Our construction is adapted from the GHS algorithm~\cite{ghs83}, a classical algorithm for constructing minimum spanning trees in the message-passing model of distributed computing. To adapt the GHS algorithm to our agent-based model, firstly, we need to assign distinct weights to the $m$ edges to meet the GHS algorithm's requirement of unique edge weights since we consider the graph of unweighted edges. This goal can also be achieved by assigning random weights to the edges, however; since our algorithm is deterministic, we use a deterministic method for assigning the edge weights. To achieve this, we use the agents' unique IDs and the port numbering at each node. Secondly, the weight assignment approach would require additional memory for storing and managing these assigned weights. For an agent $r_u$, it would be $O(\delta_u \log \lambda)$ bits, where $\delta_u$ is the degree of the node $u$. To avoid this overhead, we do not store the weights of each edge in the agent's memory but calculate these weights on the fly when required. Now, to assign a differentiating edge weight to each edge, the agents employ the following method.

Now, we present an algorithm that enables agents to construct a spanning tree for a given graph $G(V,E)$. Our approach is an adaptation of the GHS algorithm~\cite{ghs83}, a classical algorithm for constructing a Minimum Spanning Tree (MST) in the message-passing model of distributed computing. To adapt the GHS algorithm to our agent-based model, we first address the requirement of unique edge weights, as the algorithm assumes distinct weights for all $m$ edges. Since we consider a graph with unweighted edges, we must assign unique weights to meet this requirement. While this could be achieved by assigning random weights to the edges, our algorithm is deterministic, necessitating a deterministic method for weight assignment. To this end, we leverage the agents' unique identifiers (IDs) and the port numbering at each node to systematically assign distinct edge weights.Furthermore, this weight assignment approach introduces additional memory requirements for storing and managing the assigned weights. For an agent $r_u$, the memory required would be $O(\delta_u \log \lambda)$ bits, where $\delta_u$ denotes the degree of node $u$ and $\lambda$ represents the maximum ID value of the agents. To eliminate this overhead, we do not store edge weights in the agents’ memory. Instead, these weights are computed dynamically when needed. Below we outline the method used by the agents to dynamically assign distinct weights to each edge in the graph as needed.

\subsubsection{Generating Distinct Edge Weights:}\label{weights} To obtain a distinct edge weight for the edges, an agent $r_u$ does the following: 
\begin{itemize}
    \item $r_u$ meets its neighbors one by one. 
    \item Once $r_u$ meets with its neighbor $r_v$, the weight of the edge $(r_u,r_v)$ is evaluated as $r.ID+\frac{1}{p_{uv}+2}$, where $r.ID=\min{\{r_u.ID,r_v.ID\}}$ and $p_{uv}$ is the port number at agent $r$ leading to the other agent.
    \item Both $r_u$ and $r_v$ assign the same weight $r.ID+\frac{1}{p_{uv}+2}$ to the edge $(r_u,r_v)$.
\end{itemize}
%For example, let the edge $e$ connect the agents $r_3$ and $r_7$ (with resp. IDs $3$ and $7$). Let the local port number connecting $e$ at $r_3$ and $r_7$ be $2$ and $29$ respectively. Now, when the agents $r_3$ and $r_7$ meet, they both assign the weight $3+\frac{1}{2+2}=3.25$ to the edge $e$. Through this process, the agents can calculate the edge weights of the edges adjacent to it, as required. We now show the weights assigned by the above method is unique.\manish{Explantion seems reasonable to me. I think we may skip the example.}
\begin{lemma}[Distinct Edge Weights]
    The edge weights assigned to each edge are distinct.
\end{lemma}
\begin{proof}
    To prove this, we consider two arbitrary edges of $G$, say $e_1$ and $e_2$. First, we assume that $e_1$ and $e_2$ are adjacent edges having the common node $u$ containing the agent $r_u$. Let $e_1$ connect $r_u$ (at node $u$) to $r_v$ (at node $v$) and $e_2$ connect $r_u$ to $r_w$ (node $w$) via port number $p_{uv}$ and $p_{uw}$ respectively. Now, if $r_u.ID=\min\{r_u.ID,r_v.ID,r_w.ID\}$, then, the weights of $e_1$ and $e_2$ namely $w_{e_1}$ and $w_{e_2}$ are given by $r_u.ID+1/(p_{uv}+2)$ and  $r_u.ID+1/(p_{uw}+2)$ which are, clearly, distinct as $p_{uv}\neq p_{uw}$. If $r_u.ID$ is not the minimum then the weights $w_{e_1}$ and $w_{e_2}$ are assigned by two different agents from the set $\{r_u,r_v,r_w\}$. Since the IDs of these agents are integers and the addition fraction in the weight calculation due to port number i.e., $1/(port+2)$ is less than 1, so $w_{e_1}$ and $w_{e_2}$ must be distinct. On the other hand, if $e_1$ and $e_2$ are not adjacent, then $w_{e_1}$ and $w_{e_2}$ are assigned weights based on different agents' IDs, which similarly, must be distinct.
\end{proof}

%Now, to calculate the weight of an edge, the agents need to meet its neighbor once, which it can do through the \texttt{Meeting with Neighbor()} protocol. In this way, each agent can meet all its neighbors and compute the outgoing edge weights. With distinct edge weights established, the agents can now simulate the GHS algorithm to construct a spanning tree. We now describe the execution of the algorithm.

To compute the weight of an edge, an agent must interact with its neighboring agent at least once, which can be achieved using the \texttt{Meeting with Neighbor()} protocol. Through this process, each agent sequentially meets all its neighbors and determines the weights of its outgoing edges. Once distinct edge weights have been assigned, the agents proceed to simulate the GHS algorithm to construct a spanning tree. Now, we provide the details of the MST construction (Algorithm~\ref{alg: MST}). We divide the algorithm into the following three main stages:

\begin{enumerate}
    \item \textbf{Minimum Weight Outgoing Edge (MWOE) Calculation:}  
    Each fragment's leader determines the minimum weight outgoing edge (MWOE) among all its outgoing edges using Algorithm~\ref{alg: MWOE Calculation by Leader}.
    \item \textbf{MWOE Identification and Selection:}  
    Using Algorithm~\ref{alg: MWOE Selection}, an agent gets to know whether its outgoing edge is MWOE. For this, the fragment's leader broadcasts the minimum weight value to all the agents in its fragment. Once an agent realizes that one of its incident edges is the MWOE, it proceeds to the next stage.
    
    \item \textbf{Fragment Merging:}  
    In this stage, as the agent holding the MWOE (say, $r_u$) meets with its adjacent agent (say, $r_v$), $r_u$ shares its fragment level ($treelevel$) and its intention to merge with or absorb the fragment of $r_v$. At this point, agent $r_u$ executes either \texttt{Merge()} (Algorithm~\ref{merge}) or \texttt{Absorb()} (Algorithm~\ref{absorb}), depending on the respective $treelevel$ values of agents $r_u$ and $r_v$ and whether the edge $(r_u,r_v)$ is the MWOE for both fragments.
\end{enumerate}

The process continues until all the agents become part of a single fragment with no more outgoing edges. At this point, the algorithm terminates. Below, we explain the three stages in detail. 
\subsubsection{Stage $1$ - Minimum Weight Outgoing Edge (MWOE) Calculation:}

Each edge is assigned a unique virtual weight, as derived from Section~\ref{weights}. Initially, each of the $n$ agents individually represents a distinct tree fragment. As the algorithm progresses, these fragments merge via the Minimum Weight Outgoing Edge (MWOE). In Algorithm~\ref{alg: MWOE Calculation by Leader}, each fragment's leader\footnote{During MST construction, the leader is assumed to be the fragment leader unless specified otherwise. Once all fragments merge into a single one, the fragment's leader becomes the leader of all agents.} computes the MWOE for its fragment. The algorithm begins with each agent exploring its neighbors in parallel. For each neighbor, the agent checks whether the edge connecting to it is internal or outgoing by comparing their respective $treelabel$ values. If the $treelabel$ values match, the edge is internal; otherwise, it is considered outgoing (Algorithm~\ref{alg: MWOE Calculation by Leader}, Line 3). Then each agent sequentially examines its neighbors to identify the minimum-weight outgoing edge. During this process, it maintains a variable $r_u.\text{min\_weight}$ to store the smallest weight encountered among all outgoing edges. Although each agent searches its local neighborhood sequentially, these searches occur in parallel across the agents. The agent then aggregates the minimum weights received from its children, storing the local minimum in the variable $min\_{local}$ (Lines 2–11, Algorithm~\ref{alg: MWOE Calculation by Leader}). This value is subsequently propagated up in the fragment’s hierarchy with the parent node, ensuring that the fragment leader obtains the global minimum (Algorithm~\ref{alg: MWOE Calculation by Leader}, Lines 12-15). The algorithm terminates when only a single fragment remains, which occurs when there are no outgoing edges. This termination condition is detected in Algorithm~\ref{alg: MWOE Calculation by Leader}, Line 16.  Each agent is also equipped with the variable $treelevel$, which tracks the level of a given fragment. Specifically, when two fragments with the same $treelevel$ merge, the newly formed fragment's $treelevel$ is incremented. The variables $treelabel$ and $treelevel$ serve distinct roles: $treelabel$ identifies agents within the same fragment (which serves like a tree ID), while $treelevel$ facilitates the merging of different fragments. Initially, $treelevel = 0$, and $treelabel$ is set to the fragment leader's ID. Since each agent starts as an independent fragment, its $treelabel$ is initially assigned as its own ID. As the fragments merge, the $treelabel$ of the agents in the merged fragment is determined by the leader of the new fragment. 

\begin{algorithm}[htbp]
\caption{MWOE Calculation by Fragment Leader}\label{alg: MWOE Calculation by Leader}
\begin{algorithmic}[1]
% \If{$r_u$ receives a visit from its parent (when exists) to initiate MWOE computation or $r_u$ is a leader with no children}
    % \State $r_u.\text{min\_weight} \leftarrow \phi$
    % %\State \hl{$r_u$ visits each neighbor $r$ one by one and performs the following.}
    % \For{each $r_u$ (in parallel) visits each neighbor $r$}
    %     \If{$r.\text{treelabel} \neq r_u.\text{treelabel}$} \Comment{Outgoing edge check.}
    %         \State $r_u$ computes the weight $w_{r_u,r}$ of edge $(r_u, r)$.
    %     \Else
    %         \State $w_{r_u,r} \leftarrow \perp$. 
    %         % \State $r_u$ visits $r$ to and instructs to explore its children for MWOE discovery.
    %     \EndIf
        
    %     \If{$r_u.\text{min\_weight} > w_{r_u,r}$}
    %         \State $r_u.\text{min\_weight} \leftarrow w_{r_u,r}$
    %     \EndIf
    % \EndFor
% \EndIf

\State $r_u.\text{min\_weight} \leftarrow \phi$
\For{each agent $r_u$ (in parallel), and for each neighbor $r$ of $r_u$}
    \If{$r.\text{treelabel} \neq r_u.\text{treelabel}$} \Comment{Outgoing edge check}
        \State $r_u$ computes the weight $w_{r_u,r}$ of edge $(r_u, r)$
    \Else
        \State $w_{r_u,r} \leftarrow \perp$
    \EndIf
    \If{$r_u.\text{min\_weight} > w_{r_u,r}$}
        \State $r_u.\text{min\_weight} \leftarrow w_{r_u,r}$
    \EndIf
\EndFor

\State $r_u$ waits until it receives $r_v.\text{min\_weight}$ from each child $r_v$ (if any)
\State $r_u$ calculates $r_u.\text{min\_local} = \min \left( \{ r_u.\text{min\_weight} \}, \{ r_v.\text{min\_weight} \mid r_v \text{ is a child of } r_u \} \right)$ and sends it to its parent
\If{$r_u$ is the leader agent}
    \State $r_u.\text{leader\_min} \leftarrow r_u.\text{min\_local}$

\If{$r_u.\text{leader\_min} = \perp$ }
    \State \textbf{terminate algorithm.}
\EndIf
\EndIf
\end{algorithmic}
\end{algorithm}

Before proceeding to the selection of MWOE, we highlight two important remarks. The first deals with memory complexity, which enables the agents to construct the MST using $O(\log \lambda)$ bits. The second introduces a simple synchronization technique that allows the leader of a fragment to efficiently communicate with all agents in its subtree.
\begin{remark}[Memory Complexity]\label{rem:memory}
    The primary memory overhead of the algorithm arises from maintaining child pointers. An agent with up to $\Delta$ children would typically require $O(\Delta \log \lambda)$ bits for storage. However, this complexity can be reduced using a \texttt{sibling} variable. When an agent discovers its first child, it sets its \texttt{child} pointer to the corresponding outgoing port. For each subsequent child, instead of storing multiple child pointers, the agent updates the previously added child's \texttt{sibling} variable with the new child's outgoing port. This cascading approach eliminates the need for $\Delta$ child pointers, allowing the \texttt{child} pointers to be stored efficiently in $O(\log \lambda)$ bits. The other variables \texttt{parent}, \texttt{child}, and \texttt{treelabel} require $O(\log \lambda) + O(\log \Delta) = O(\log \lambda)$ bits, while other variables require only $O(1)$ bits. Thus, the overall memory complexity is $O(\log \lambda)$ or $ O(\log n)$ bits, since $\lambda\leq n^{O(1)}$.
\end{remark}

% \begin{remark}[\texttt{parent-child} communication]
%     To efficiently exchange information between a parent and its children in a single round, we employ the following strategy: Throughout the algorithm, after every $4\log\lambda$ steps of neighbor interactions, an additional step is reserved where each agent (except the leader) moves to its parent to check for updates. Thus, after every $4\log\lambda$ rounds, each child oscillates once between its position and its parent.  Now, suppose the leader wants to broadcast information to all agents. In the first oscillation, its immediate children receive the value. In the next, the grandchildren obtain it, and this process continues. After $n$ such oscillations, all agents receive the value. However, there is still a small tweak required. If the agents continue to oscillate in periodic intervals, the agnets may never receive any updated values from the leader. So, we add a modification. Once, the agnets who are direct children of the leader gets updated value, they stop movement in the next period. This allows the agents in the next level (grandchildren of leader) to access this updated value. This allows the smooth trander of information from parent to children. Conversely, if children need to send information to their parent, they can move simultaneously during this reserved step. Thus, transmitting any value (e.g., computing a minimum) requires an additional $D$ rounds of communication. Here $D$ is the diameter of the constructed tree, which in the worst case, could be $O(n)$.
% \end{remark}

\begin{remark}[\texttt{parent-child} communication]\label{rem:oscilate}
To efficiently exchange information between a parent and its children in a single round, we introduce the following strategy: After every $4\log\lambda$ rounds of neighbor interactions, an additional round is reserved in which each agent (except the leader) moves to its parent for updates. Thus, after every $4\log\lambda$ round, each child oscillates once between its node and its parent node.  When the leader initiates a broadcast, its immediate children receive the value in the first oscillation, followed by the grandchildren in the next, continuing until all agents are updated after $D$ oscillations, where $D$ is the tree diameter. To ensure updates propagate correctly, once all children at a given level have received the update, they stay at their node for $4\log \lambda$ rounds, allowing deeper levels to update sequentially. This ensures a smooth transfer of information.  Conversely, during the same reserved round, children can send information to their parents simultaneously. Thus, transmitting any value requires an additional $O(n)$ rounds in the worst case.
\end{remark}

 % 

% \begin{algorithm}[t]
% \caption{MWOE Calculation by Fragment Leader}\label{alg: MWOE Calculation by Leader}
% \begin{algorithmic}[1]
% \If{\hl{$r_u$ receives a visit from its parent (when exists) to initiate MWOE computation or $r_u$ is a leader with no children}}
%     \State $r_u.\text{min\_weight} \leftarrow \phi$
%     \State \hl{$r_u$ visits each neighbor $r$ one by one and performs the following.}
%     \For{each neighbor $r$ of $r_u$}
%         \If{$r.\text{treelabel} \neq r_u.\text{treelabel}$} \Comment{Outgoing edge check}
%             \State $r_u$ computes the weight $w_{r_u,r}$ of edge $(r_u, r)$.
%         \Else
%             \State $w_{r_u,r} \leftarrow \perp$. 
%             \State $r_u$ visits $r$ to and instructs to explore its children for MWOE discovery \manish{Is this step necessary?}\prabhat{Possibly. Particularly, when the tree grows bigger. }
%         \EndIf
        
%         \If{$r_u.\text{min\_weight} > w_{r_u,r}$}
%             \State $r_u.\text{min\_weight} \leftarrow w_{r_u,r}$
%         \EndIf
%     \EndFor
% \EndIf
% \State $r_u$ waits until it receives $r_v.\text{min\_weight}$ from each child $r_v$ (if any)
% \State $r_u$ calculates $r_u.\text{min\_local} = \min \left( \{ r_u.\text{min\_weight} \}, \{ r_v.\text{min\_weight} \mid r_v \text{ is a child of } r_u \} \right)$ and sends it to its parent
% \If{$r_u$ is the leader agent}
%     \State $r_u.\text{leader\_min} \leftarrow r_u.\text{min\_local}$

% \If{$r_u.\text{leader\_min} = \perp$ }
%     \State \textbf{terminate algorithm.}
% \EndIf
% \EndIf
% \end{algorithmic}
% \end{algorithm}

\subsubsection{Stage $2$ - MWOE Identification and Selection:} To determine the MWOE of a fragment, an agent must verify whether any of its adjacent edges has been selected by the leader as the MWOE for connecting to another fragment. This requires the leader to broadcast the minimum computed weight value to all agents, a process carried out using Algorithm~\ref{alg: MWOE Selection}. Each agent checks if any of its adjacent edges is the MWOE by comparing its variable $r_u.min\_weight$ with the value sent by the leader; if they match, one of its adjacent edges is identified as the MWOE (Line 6). Since agents do not store edge weights in memory, they must reevaluate the weights (Line 7) to identify the edge corresponding to the MWOE weight received from the leader, introducing an additional computational factor of $\Delta \log n$ into Algorithm~\ref{alg: MWOE Calculation by Leader}. With the MWOE of the fragment determined, we now proceed to describe the merging of fragments.

\begin{algorithm}[htbp]
\caption{MWOE Identification and Selection}\label{alg: MWOE Selection}
\begin{algorithmic}[1]
    \State $r_u$ checks $r_t.\text{leader\_min}$.\Comment{Communication within the constructed tree is achieved using techniques established in Remark~\ref{rem:oscilate}}
    \If{$r_u.\text{leader\_min} \neq r_t.\text{leader\_min}$} \Comment{Fragment leader has calculated a new weight}
    \State $r_u.\text{leader\_min} \leftarrow r_t.\text{leader\_min}$
    \If{$r_u.\text{min\_weight} = r_u.\text{leader\_min}$} \Comment{One of $r_u$'s adjacent edges is the MWOE}
        \State $r_u$ visits each neighbor to re-calculate edge weights and find the edge matching $r_u.\text{leader\_min}$
        \State The matched edge is computed as the MWOE, connecting $r_u$ to some neighbor $r_k$ (say)
        \State $ r_u $ informs $ r_k $ that its incident edge is the MWOE and schedules it for merging.

    \EndIf
\Else \Comment{No new edge weight calculated}
    \State $r_u$ does nothing
\EndIf
\end{algorithmic}
\end{algorithm}

% \begin{algorithm}[t]
% \caption{MWOE Identification and Selection}\label{alg: MWOE Selection}
% \begin{algorithmic}[1]
%     \State $r_u$ checks $r_t.\text{leader\_min}$.\Comment{Communication within the constructed tree is achieved using techniques established in Remark~\ref{rem:oscilate}}
%     \If{$r_u.\text{leader\_min} \neq r_t.\text{leader\_min}$} \Comment{Fragment leader has calculated a new weight}
%     \State $r_u.\text{leader\_min} \leftarrow r_t.\text{leader\_min}$
%     \If{$r_u.\text{min\_weight} = r_u.\text{leader\_min}$} \Comment{One of $r_u$'s adjacent edges is the MWOE}
%         \State $r_u$ visits each neighbor to re-calculate edge weights and find the edge matching $r_u.\text{leader\_min}$
%         \State The matched edge is computed as the MWOE, connecting $r_u$ to some neighbor $r_k$ (say)
%         \State $ r_u $ informs $ r_k $ that its incident edge is the MWOE and schedules it for merging.

%     \EndIf
% \Else \Comment{No new edge weight calculated}
%     \State $r_u$ does nothing
% \EndIf
% \end{algorithmic}
% \end{algorithm}

\subsubsection{Stage $3$ - Fragment Merging:}
    Let $F_1$ and $F_2$ be two fragments connected via a MWOE $e$. Then the fragments $F_1$ and $F_2$ merge based on the following rule:
    \begin{enumerate}
        \item \textbf{Rule 1} - If $F_1.treelevel<F_2.treelevel$, then $F_1$ joins the fragment $F_2$ and the $treelabel$ of the new fragment will be $F_2.treelabel$. Moreover, the leader of the fragment will be the leader of fragment $F_2$.
        \item \textbf{Rule 2} - If $F_1.treelevel=F_2.treelevel$, then $F_1$ and $F_2$ will merge if and only if, $e$ is the common MWOE of $F_1$ and $F_2$. In such a case, the leader of the fragment will be the leader of $F_1$ and $F_2$ that has the minimum ID.
        \item \textbf{Rule 3} - In other cases, fragments wait till one of the rules 1 or 2 applies to them. 
    \end{enumerate}
     Now, let us consider two adjacent agents $r_u$ and $r_k$ connected via the MWOE $e$ of $r_u$. Now, $r_u$ does the following
     \begin{enumerate}
         \item If $r_u.treelevel>r_k.treelevel$, $r_u$ does nothing.
         \item If $r_u.treelevel=r_k.treelevel$ and $e$ is the MWOE of the fragment belonging to $r_k$ as well, $r_u$ executes \texttt{Merge()}.
         \item If $r_u.treelevel=r_k.treelevel$ and $e$ is \textbf{not} the MWOE of the fragment belonging to $r_k$, $r_u$ does nothing.
         \item If $r_u.treelevel<r_k.treelevel$, $r_u$ executes \texttt{Absorb()}.
     \end{enumerate}
 Specifically, \texttt{Merge()} is used when two fragments have the same level, requiring an increment in their $treelevel$ value. In contrast, \texttt{Absorb()} is typically applied when two fragments can be ordered based on their $treelevel$ and $ID$ values. Through \texttt{Absorb()}, the fragment of $r_u$ merges into the fragment of $r_k$. During the merging process, the merging fragment terminates all ongoing operations and updates its child and parent pointers, as well as $treelabel$ and $treelevel$ accordingly. 

\begin{algorithm}[H]
\caption{Absorb()}\label{absorb}
\begin{algorithmic}[1]
          \State $r_u$ immediately aborts any operation it may be doing.
          \State $r_u$ marks $r_k$ as parent and collects $r_k.treelabel,r_k.treelevel$ and leader information of $r_k$'s fragment. 
          \State $r_u$ moves to its parent agent, and instructs to immediately abort any ongoing operation.
          \State $r_u$ informs its parent about the new leader, changes its parent's $treelabel,treelevel$ to $r_k.treelabel,r_k.treelevel$, and changes the parent pointer to child pointer.
          \State Once the parent pointer has been changed, the old parent moves to its own parent, requests abortion of any ongoing operation, marks its parent pointer as a child pointer, and similarly transmits and modifies the new information. Continuing similarly, when this message reaches the root, the root sends out completion information to $r_k$ via the newly changed parent chain. 
          \State Once $r_u$, receives the completion information from its old root, it sends this new leader and $treelabel,treelevel$ information across its other (old) children in a similar fashion. The leaf children, after modifying their own information, again send out completion information to $r_u$. After $r_u$ has received the completion information from all of its children, it now informs $r_k$ about the completion of the \texttt{ Absorb()}. As $r_u$ visits $r_k$ with the completion information, $r_k$ sets the port leading to $r_u$ as its child pointer. 
          \State The fragment of $r_u$ has now merged with $r_k$'s fragment with a new leader of $r_k$'s fragment.
\end{algorithmic}
\end{algorithm}

\vspace{-0.75cm}

\begin{algorithm}[H]
\caption{Merge()}\label{merge}
\begin{algorithmic}[1]
\State $r_u$ and $r_k$ compare their leader IDs
\If{$r_u$ has a higher leader ID}
    \State $r_k$ increments $r_k.treelevel$ by $1$.
    \State $r_k$ executes \texttt{Absorb()}
\Else
    \State $r_u$ increments $r_k.treelevel$ by $1$.
    \State $r_u$ executes \texttt{ Absorb()}
\EndIf
\end{algorithmic}
\end{algorithm}

\begin{algorithm}[H]
\caption{MST Construction Using Mobile Agents}\label{alg: MST}
\begin{algorithmic}[1]
  \While{$r_u.leader\_min\neq \perp$}
    \State Calculate the MWOE using Algorithm~\ref{alg: MWOE Calculation by Leader}.
    \State Identify and select the MWOE using Algorithm~\ref{alg: MWOE Selection}.\Comment{Let $e$ be the MWOE from the agent $r_u$ which connects to the agent $r_k$.}
        \If{$r_u.treelevel=r_k.treelevel$ and $r_k$ also scheduled for merging via $e$} 
        \State Execute \texttt{Merge()}.
        \ElsIf{$r_u.treelevel=r_k.treelevel$ and $r_k$ not scheduled for merging via $e$}
        \State $r_u$ does nothing and waits. 
        \ElsIf{$r_u.treelevel>r_k.treelevel$}
        \State $r_u$ does nothing and waits.
        \ElsIf{$r_u.treelevel<r_k.treelevel$}
        \State $r_u$ executes \texttt{Absorb()}. 
        \EndIf
  \EndWhile
\end{algorithmic}
\end{algorithm}

\subsubsection{Complexity Analysis of Algorithm~\ref{alg: MST} (MST Construction)}
\begin{lemma}
    The respective time complexities of the 3 stages of the MST construction are respectively $O(\Delta\log\lambda+n)$, $(\Delta\log\lambda+n)$ and $O(n)$ rounds.
\end{lemma}
\begin{proof}
    For \textit{Stage 1}, each agent takes $O(\Delta\log\lambda)$ rounds in the worst case to search its neighbors, which is done in parallel by the agents of a fragment. To report the minimum weight to the leader, the agents must send this value via their parent, and the leader computes the minimum of all the received weights. Since this transmission occurs via the tree, therefore in the worst case it can take $O(n)$ time. So, calculating the MWOE by the leader takes $O(\Delta\log\lambda+n)$ rounds. To distribute the MWOE to the whole fragment, the agents can take a maximum of $O(n)$ rounds as the agents that are at the same distance from the root parallelly collect this MWOE from their parents. Now the agents once again evaluate the edge weights to compute the matching edge with the selected edge weight. So, \textit{Stage 2} also takes $O(n+\Delta\log\lambda)$ rounds. Finally, for \textit{Stage 3}, the agent that connects its fragment via the MWOE needs to visit the chain of parents to inform them about the merging and changing of their child-parent pointers, respectively. This modification takes $O(n)$ rounds in the worst case.
\end{proof}

According to the GHS algorithm, the number of fragments while constructing the spanning tree reduces by at least half in each phase, so it takes $O(\log n)$ phases for the fragments to merge into a single fragment, giving out the final spanning tree. Now since, during each phase, weight is evaluated (at most twice), Algorithm~\ref{alg: MWOE Selection} and Algorithm~\ref{alg: MWOE Calculation by Leader} is invoked once followed by merging of fragments, the time complexity of constructing the spanning tree is given by $O(\log n)\cdot(O(\Delta\log \lambda)+O(\Delta\log \lambda+n)+O(n+\Delta\log \lambda)+O(n))=O(\Delta\log n \log \lambda+n\log n)$ rounds. Considering that $O(\log\lambda)=O(\log n)$, we can rewrite the complexity as $O(\Delta\log^2 n+n\log n)$. The maximum memory overhead for the agents is used to store the IDs and pointers. Since the IDs are bound by $n^c$ for some constant $c\geq1$, the memory requirement for each agent is $O(\log n)$ bits. Combining other details from the discussion in Remark~\ref{rem:memory}, the correctness of the original GHS algorithm, and the preceding lemma, we have the following theorem.
\begin{theorem}[MST]\label{theorem: MST}
    Given an arbitrary simple connected graph $G$ with $n$ nodes $m$ edges, maximum degree $\Delta$ and diameter $D$. Then, an MST can be constructed in  $O(\Delta\log^2 n+n\log n)$ rounds, when the $n$ agents begin in a \emph{dispersed} initial configuration with $O(\log n)$ bits of memory per agent. The agents do not require any prior knowledge about any graph parameters but need to know an upper bound of the ID of the agents, $\lambda$. 
\end{theorem}

\subsection{From the Unweighted-MST to Leader Election and Weighted-MST}

    Following the construction of the spanning tree, all fragments successfully merge into a single component, resulting in a unique leader election that is part of a single fragment having $n$ agents. Consequently, we establish the following \textbf{Theorem~\ref{theorem: Leader}}, which guarantees the existence of a single leader within the graph.

    \begin{theorem}[Leader Election]\label{theorem: Leader}
     Given a dispersed configuration of $n$ agents with unique identifiers positioned one agent at every node of $n$ nodes $m$ edges, maximum degree $\Delta$ graph $G$ with no node identifier. Then, there is a deterministic algorithm that elects a leader in $O(\Delta \log^2 n+n\log n)$ rounds and $O(\log n)$ bits per agent. The algorithm requires no prior knowledge of any graph parameters and assumes that each agent knows only $\lambda$, the maximum among all agent identifiers.
\end{theorem}

    In the case of a weighted graph with distinct edge weights, the algorithm proceeds analogously to \textbf{Algorithm~\ref{alg: MST}}, with the primary distinction being that edge weights are predefined and do not require explicit computation (as discussed in \textbf{Section~\ref{weights}}). Consequently, all analytical results, including time and memory complexity, remain consistent with those derived for \textbf{Algorithm~\ref{alg: MST}}. This leads to the formulation of \textbf{Theorem~\ref{theorem: MST}}.
    
    With this knowledge, we proceed to the construction of the BFS tree of the graph using agents (discussed in the next section). The BFS tree construction requires that all agents be aware of the values of  $n$, $m$, and $\Delta$ beforehand. Utilizing the spanning tree structure, agents compute fundamental graph parameters—number of nodes ($n$), number of edges ($m$), and maximum degree ($\Delta$)—through data aggregation and maximum value identification. This process follows a methodology similar to the Minimum Weight Outgoing Edge (MWOE) computation and dissemination, as described in Algorithm~\ref{alg: MWOE Calculation by Leader} (aggregation) and Algorithm~\ref{alg: MWOE Selection} (spreading). Each agent transmits its locally computed values to its parent, accumulating the required information at the MST leader within $O(n)$ rounds. The leader then consolidates these values and propagates them back across the spanning tree, ensuring all agents acquire a global view of the network parameters.

\section{Agent-Based BFS Tree Construction}\label{sec: BFS_improved}

The paper~\cite{CKM24} presents a BFS construction algorithm that requires $O(D\Delta)$ rounds when a root is specified. In this section, we introduce a new BFS construction algorithm that completes in $O(m \log n)$ rounds, given a known root node and the number of edges $m$. By combining these two algorithms, we achieve a round complexity of $O(\min(D\Delta, m \log n))$ for the BFS construction. The assumptions of a known root and $m$ can be removed using the leader election protocol in Section~\ref{sec: MST_construction}. In fact, the values of $m$, and any other necessary graph parameters can be computed by Algorithm~\ref{alg: MST}. The BFS algorithm heavily uses the {\texttt{Meeting with Neighbor()} (Algorithm~\ref{alg:meeting_protocol})} as a subroutine, which requires $\lambda$ to be known to the agents. Below, we discuss an overview of the algorithm.

\medskip
\noindent {\bf High-Level Idea:} The algorithm proceeds by propagating levels and marking tree edges in the BFS tree, similar to the flooding-based algorithm in the message-passing model. We use the {\texttt{Meeting with Neighbor()} (Algorithm~\ref{alg:meeting_protocol})} to send level information from one agent to its neighboring agents, one by one. Suppose an agent $r^\star$ is located at the specified root node, which is at level $0$ in the BFS tree. We denote an agent's level as the level of its corresponding node in the BFS tree. Thus, the root agent $r^\star$ sets its level as $r^\star.level = 0$, and every other agent $r_u$ sets its level as $r_u.level = \perp$ initially. Whenever an agent $r_u$ updates its level, it informs its neighboring agents of its new level. A neighboring agent $r_v$ updates its level such that $r_v.level = r_u.level+1$ if, $r_v.level> r_u.level+1$ or $r_v.level= \perp$. $r_v$ considers $r_u$ as its parent and the corresponding edge is considered as the edge of the BFS. In case more than one neighbor sends their level simultaneously (in the same round) such that they can be considered as the parent, then the minimum level neighbor is considered as the parent. In case, there is more than one such minimum considerable level send their level then one of them is arbitrarily considered as the parent. The agent informs its neighbors one by one based on port numbering (in increasing order) using the {\texttt{Meeting with Neighbor()}} that takes $O(4 \log \lambda)$. Notice that to meet any neighbor {\texttt{Meeting with Neighbor()}} start the procedure in the round number that is divisible by $4\log \lambda$ and this procedure runs for $2m$ times, after which each agent learns its actual level in the BFS tree. %The correctness of the algorithm is follows from the Lemma~\ref{lem: BFS_correctness}. The pseudocode is given in Algorithm~\ref{alg: BFS_another_approach}, and the complexity analysis of the algorithm is provided in Lemma~\ref{lem: BFS_another_approach}. 

Let us now briefly explain the $O(D\Delta)$ rounds BFS tree construction algorithm. Chand {\it et al.} constructed the BFS tree in $O(D\Delta)$ rounds and $O(\log n)$ bits of memory per agent in~\cite{CKM24}. For the BFS construction, they assumed the dispersed initial configuration and prior knowledge of the root node. For that, first, the algorithm computed the $\Delta$ in $O(D\Delta)$ rounds, and then each level took $O(\Delta)$ rounds to complete a level during the BFS tree construction. Thus, their algorithm required $O(D\Delta)$ rounds. Their approach required $O(\log n)$ bits of memory per agent since each agent stored only a constant number of level-numbers, parent, ID, $\Delta$, etc. After constructing BFS, with the help of a parent port, each agent informed its parent, and eventually, the root got to know about the BFS formation, and the algorithm terminated. For the detailed analysis, refer to~\cite{CKM24}.

Now, we have two algorithms with round complexity $O(m\log n)$ (from Algorithm~\ref{alg: BFS_another_approach}) and $O(D\Delta)$ (from~\cite{CKM24}). We consider the best of these two algorithms for BFS construction. For that, first, we run the $O(D\Delta)$ round algorithm; either that algorithm terminates within $O(m \log n)$ rounds if so then we have the BFS tree. Otherwise, we run the Algorithm~\ref{alg: BFS_another_approach} that terminates in $O(m \log n)$. Thus, we get the best round complexity of these two algorithms. The pseudocode for that is given in Algorithm~\ref{alg: improved_BFS}.%We provide the pseudocode of the combined algorithms for the BFS construction in Algorithm~\ref{alg: improved_BFS} which achieves the desired result.

\begin{algorithm}[ht]
\caption{$O(m \log n)$-Round BFS Construction.}\label{alg: BFS_another_approach}
\begin{algorithmic}[1]
\Require Graph $G$ with a specified root and $m$. % and $\lambda$ in a dispersed configuration of the agents. 
\Ensure BFS tree where each tree edge is identified by an agent. 

\State Let us consider the agent $r_u$ whose level is denoted by $r_u.level$. Initially, every $r_u$ agent has $r_u.level=\perp$ and $r_u.parent \xleftarrow{} \perp$.
\State Consider root agent $r^\star$ such that $r^\star.level=0$.
\State  $r^\star$ sends $r^\star.level$ to its neighbors one-by-one using the meeting neighbors algorithm {\texttt{Meeting with Neighbor()} (Algorithm~\ref{alg:meeting_protocol})}. \label{line: root_level_passing} \Comment{Takes $4 \log \lambda$ rounds}
\For{$2m$ times}
    \If{agent $r_u$ updates its level}\label{line: level_passing}
        \State $r_u$ sends $r_u.level$ to neighbors one-by-one using the {\texttt{Meeting with Neighbor()}}. %]\anis{after updating, will the agent immediately start sending new level to the neighbors? If yes, will \texttt{Meeting with Neighbor()} take care of sync? } 
        \Comment{Takes $4 \log \lambda$ rounds}
    \EndIf
    
   %Let $\ell = \min \{ L(r_v) \mid r_v \in P \}$ where $L(r_v)$ represents the level of a neighbor $r_v$ in the set $P$.

    Let  $\ell = \min \{ L(r_v) \mid r_v \in P \}$  where $L(r_v)$ denotes the level of a neighbor $r_v$ in the set $P$.  

    % \If{$r_u$ receives $\ell$ such that $\ell+1 < r_u.level$ or $r_u.level = \perp$}
    %     \State $r_u.level =  \ell +1$
    %     \State $r_u.parent \leftarrow r_v$ such that $r_v \in P$ where $r_v$ is the arbitrarily chosen neighbor.
    % \EndIf

     \If{$r_u$ receives $\ell$ such that $(\ell + 1 < r_u.level)$ or $(r_u.level = \perp)$}
        \State $r_u.level \gets \ell + 1$
        \State Choose an arbitrary $r_v \in P$ such that $L(r_v) = \ell$
        \State $r_u.parent \gets r_v$
    \EndIf
    
\EndFor
\end{algorithmic}
\end{algorithm}

\begin{algorithm}[htbp]
\caption{Improved BFS Construction.}\label{alg: improved_BFS}
\begin{algorithmic}[1]
\Require Graph $G$ having $n$-nodes and $n$-agents with dispersed configuration.
\Ensure BFS Construction

\State Elect the leader/root by running the MST Algorithm~\ref{alg: MST} and knows the graph parameter, $m, \lambda$ and $\Delta$ in $O(n)$ rounds. (see Section~\ref{sec: MST_construction}).
%\State Root agent, $r^\star$,  from leader election (See~\ref{leader_election}).
\State Root, $r^\star$, runs the BFS algorithm for $8m\log \lambda$ rounds from~\cite{CKM24}.
%\State Each agent $r_v$ informs its neighbor whether $r_v$ got its level. \Comment{Requires $O(\Delta \log \lambda)$ rounds.}
\If{$r^\star$ is unaware of the BFS construction after $8m\log \lambda$}
    \State $r^\star$ runs the Algorithm~\ref{alg: BFS_another_approach} and construct the BFS.
\EndIf
\end{algorithmic}
\end{algorithm}

%\subsection{Proofs of Section~\ref{sec: BFS_improved}}\label{sec: appendix_bfs_proof}
\begin{lemma}\label{lem: BFS_correctness}
     Algorithm~\ref{alg: BFS_another_approach} correctly constructs a BFS tree.
\end{lemma}
\begin{proof}
    We prove two cases to show that Algorithm~\ref{alg: BFS_another_approach} constructs a BFS. i) Each agent $r_v$ is assigned with a minimum level (shortest distance from the root) ii) there does not exist a cycle.
    We prove these cases by contradiction. Let us assume that the level assigned to an agent $r_v$ is not the minimum then in the Algorithm~\ref{alg: BFS_another_approach}, if the level assigned to an agent $r_v$ is not minimum then there exists a neighbor of $r_v$, say $r_u$, such that $r_u.level+1 < r_v.level$ or $r_v.level = \perp$. In both conditions, the neighbor $r_u$ did not pass its level to $r_v$ which is contradictory to the Line~\ref{line: level_passing}. If $r_u$ also does not have the minimum level then one of the neighbors did not pass its level to $r_u$ and so on. Eventually, by induction, we have that $r^\star$ does not pass its level to one of its neighbors which is a contradiction to the Line~\ref{line: root_level_passing}. 
    
    For the second case, let us consider agents $r_1, r_2, r_3, \dots,  r_w$ forms a cycle such that (w.l.g.) $r_2.parent \xleftarrow{}r_1, r_3.parent \xleftarrow{}r_2$, \dots, $r_1.parent \xleftarrow{}r_w$. This implies $r_1.level>r_2.level>r_3.level>\cdots >r_w.level > r_1.level$. This implies there exists more than one parent to an agent which is contradictory.
\end{proof}
\begin{lemma}\label{lem: BFS_another_approach}
    Algorithm~\ref{alg: BFS_another_approach}, constructs a BFS tree in $O(m \log \lambda)$ rounds and requires $O(\log n)$ bits of memory at each agent.
\end{lemma}

\begin{proof}
    Let us assume there exists an agent $r_v$ that requires $8m\log \lambda + t$ rounds to get the minimum level, for any $t>0$. $r_v.parent$'s, say $r_u$, took at most $4\log \lambda \cdot \delta_u$ rounds to send the minimum level to $r_v$. Similarly, $r_u.parent$'s, say $r_w$, took at most $4\log \lambda \cdot \delta_w$ rounds to send the minimum level to $r_u$ and so on. Therefore, the total number of edges, in the shortest path, from $r^\star$ to $r_v$ is not more than the sum of the degree of $r_v.parent$'s and their parent, recursively, up to $r^\star$. As a result, this degree sum is not more than $2m$.  Recall that communication between any two neighboring agents takes not more than $4\log \lambda$ rounds \ref{alg:meeting_protocol}. Consequently, the total number of rounds required to send the minimum level to any of the agents is not more than $8m \log \lambda$ which is a contradiction. In the case of memory per agent, the Algorithm~\ref{alg: BFS_another_approach} stores only its parent's port number and level which does not require more than $O(\log n)$ bits at each agent. Hence, the lemma.
\end{proof}

\begin{remark}
    For an arbitrary initial configuration\footnote{An initial configuration in which neither all the agents are at a single node nor it is a dispersed configuration.}, agents achieve dispersion in $O(n \log^2 n)$ rounds~\cite{sudo24}. Therefore, if $n$ is known then after $O(n\log^2 n)$ rounds we have the dispersed setting. In the dispersed setting, the round complexity for BFS construction is $O(\min(D\Delta, m \log n) + n \log n + \Delta \log^2 n)$ (as shown in Theorem~\ref{theorem: BFS}). Thus, the overall round complexity for an arbitrary initial configuration BFS construction becomes $O(\min(D\Delta, m \log n) + n \log^2 n)$ while the round complexity remains unchanged.

\end{remark}
% \begin{algorithm}[H]
% \caption{Improved BFS Construction.}\label{alg: improved_BFS}
% \begin{algorithmic}[1]
% \Require Graph $G$ having $n$-nodes and $n$-agents with dispersed configuration.
% \Ensure BFS Construction

% \State Elect the leader/root by running the MST Algorithm~\ref{alg: MST} and knows the graph parameter, $m, \lambda$ and $\Delta$ in $O(n)$ rounds. (see Section~\ref{sec: MST_construction}).
% %\State Root agent, $r^\star$,  from leader election (See~\ref{leader_election}).
% \State Root, $r^\star$, runs the BFS algorithm for $8m\log \lambda$ rounds from~\cite{CKM24}.
% %\State Each agent $r_v$ informs its neighbor whether $r_v$ got its level. \Comment{Requires $O(\Delta \log \lambda)$ rounds.}
% \If{$r^\star$ is unaware of the BFS construction after $8m\log \lambda$}
%     \State $r^\star$ runs the Algorithm~\ref{alg: BFS_another_approach} and construct the BFS.
% \EndIf
% \end{algorithmic}
% \end{algorithm}
%The agent $r\star$ informs about its level in $O(\log n)$ rounds to its every neighboring agent one by one and their level is increased by 1. The neighboring agents at level 1 consider the root as their parent. Furthermore, the agent that received their level repeats the process recursively and informs its neighboring agent of their level. If the agent at the receiver end has a higher level then it updates its level with a lower level and accordingly updates its parents otherwise ignore that agent's information.

From the above discussion and BFS construction in~\cite{CKM24}. We have the following theorem.

\begin{theorem}[BFS Tree]\label{theorem: BFS}
    Given a dispersed configuration of $n$ agents with unique identifiers positioned one agent at every node of $n$ nodes $m$ edges, maximum degree $\Delta$, and diameter $D$ graph $G$ with no node identifier. Then, there is a deterministic algorithm that constructs a BFS tree in $O(\min(D\Delta, m \log n) + n \log n + \Delta \log^2 n)$ rounds and $O(\log n)$ bits per agent. The algorithm requires no prior knowledge of any graph parameters and assumes that each agent knows only $\lambda$, the maximum among all agent identifiers.
    
\end{theorem}

% \begin{remark}
%     For an arbitrary initial configuration\footnote{An initial configuration in which neither all the agents are at a single node nor it is a dispersed configuration.}, agents achieve dispersion in $O(n \log^2 n)$ rounds~\cite{sudo24}. Therefore, if $n$ is known then after $O(n\log^2 n)$ rounds we have the dispersed setting. In the dispersed setting, the round complexity for BFS construction is $O(\min(D\Delta, m \log n) + n \log n + \Delta \log^2 n)$ (as shown in Theorem~\ref{theorem: BFS}). Thus, the overall round complexity for an arbitrary initial configuration BFS construction becomes $O(\min(D\Delta, m \log n) + n \log^2 n)$ while the round complexity remains unchanged.

% \end{remark}

\section{Conclusion and Future Works}

In this paper, we analyze the complexity of constructing a BFS tree using mobile agents, improving upon \cite{CKM24} in two key ways: (1) eliminating the need for a pre-designated root and (2) reducing the round complexity from $O(D\Delta)$ to $O(\min(D\Delta, m \log n) + n \log n + \Delta \log^2 n)$. We also develop a Minimum Spanning Tree (MST) protocol for gathering graph parameters and electing a leader for BFS root selection, significantly improving time and memory efficiency over \cite{KKMS24}.  Both constructions leverage a new neighbor-meeting protocol, allowing adjacent agents to meet in $O(\log n)$ time without global graph parameter knowledge. This eliminates the need for the $\Delta$ assumption, and this may further reduce time complexity in existing solutions.  Future work includes exploring neighbor-meeting in models where ID sizes are unknown, establishing round lower bounds for agent-based BFS, and analyzing fault tolerance under crash and Byzantine failures, given their real-world relevance.

%Future work could explore the lower bounds of these computations, particularly examining if agents can meet their neighbors in constant rounds or without using the parameter $\Lambda$. It would also be worthwhile to investigate whether BFS trees can be constructed independently of spanning trees. Furthermore, analyzing this problem in the presence of faulty agents—both crash and Byzantine faults—would be valuable, as faults are inevitable in real-world applications

\bibliographystyle{plain}%\small
\bibliography{reference}

\end{document}